\documentclass[conference,10pt,twocolumn,letter]{IEEEtran}
\IEEEoverridecommandlockouts

\usepackage{cite}
\usepackage[T1]{fontenc}
\usepackage{graphicx}
\usepackage{amssymb}
\usepackage{amsmath}
\usepackage{amsthm}

\usepackage{paralist}
\usepackage{setspace}
\usepackage{microtype}
\usepackage{hyperref}
\usepackage{url}
\usepackage{balance}
\usepackage{subfigure}
\usepackage{xcolor} 
\usepackage{fancyref}
\usepackage{framed}
\usepackage{booktabs}
\usepackage{nicefrac}
\usepackage{multirow}

\usepackage{algorithm}
\usepackage{algorithmic}
\usepackage{cancel}
\usepackage{bigints}

\usepackage{amssymb}
\usepackage{amsfonts}
\usepackage{mathrsfs}
\usepackage{xspace}
\usepackage{bm}
\usepackage{upgreek}

\newcommand{\safemath}[2]{\newcommand{#1}{\ensuremath{#2}\xspace}}

\safemath{\bma}{\mathbf{a}}
\safemath{\bmb}{\mathbf{b}}
\safemath{\bmc}{\mathbf{c}}
\safemath{\bmd}{\mathbf{d}}
\safemath{\bme}{\mathbf{e}}
\safemath{\bmf}{\mathbf{f}}
\safemath{\bmg}{\mathbf{g}}
\safemath{\bmh}{\mathbf{h}}
\safemath{\bmi}{\mathbf{i}}
\safemath{\bmj}{\mathbf{j}}
\safemath{\bmk}{\mathbf{k}}
\safemath{\bml}{\mathbf{l}}
\safemath{\bmm}{\mathbf{m}}
\safemath{\bmn}{\mathbf{n}}
\safemath{\bmo}{\mathbf{o}}
\safemath{\bmp}{\mathbf{p}}
\safemath{\bmq}{\mathbf{q}}
\safemath{\bmr}{\mathbf{r}}
\safemath{\bms}{\mathbf{s}}
\safemath{\bmt}{\mathbf{t}}
\safemath{\bmu}{\mathbf{u}}
\safemath{\bmv}{\mathbf{v}}
\safemath{\bmw}{\mathbf{w}}
\safemath{\bmx}{\mathbf{x}}
\safemath{\bmy}{\mathbf{y}}
\safemath{\bmz}{\mathbf{z}}
\safemath{\bmzero}{\mathbf{0}}
\safemath{\bmone}{\mathbf{1}}

\bmdefine{\biad}{a}
\bmdefine{\bibd}{b}
\bmdefine{\bicd}{c}
\bmdefine{\bidd}{d}
\bmdefine{\bied}{e}
\bmdefine{\bifd}{f}
\bmdefine{\bigd}{g}
\bmdefine{\bihd}{h}
\bmdefine{\biid}{i}
\bmdefine{\bijd}{j}
\bmdefine{\bikd}{k}
\bmdefine{\bild}{l}
\bmdefine{\bimd}{m}
\bmdefine{\bind}{n}
\bmdefine{\biod}{o}
\bmdefine{\bipd}{p}
\bmdefine{\biqd}{q}
\bmdefine{\bird}{r}
\bmdefine{\bisd}{s}
\bmdefine{\bitd}{t}
\bmdefine{\biud}{u}
\bmdefine{\bivd}{v}
\bmdefine{\biwd}{w}
\bmdefine{\bixd}{x}
\bmdefine{\biyd}{y}
\bmdefine{\bizd}{z}

\bmdefine{\bixid}{\xi}
\bmdefine{\bilambdad}{\lambda}
\bmdefine{\bimud}{\mu}
\bmdefine{\bithetad}{\theta}
\bmdefine{\biphid}{\phi}
\bmdefine{\bideltad}{\delta}

\safemath{\bmia}{\biad}
\safemath{\bmib}{\bibd}
\safemath{\bmic}{\bicd}
\safemath{\bmid}{\bidd}
\safemath{\bmie}{\bied}
\safemath{\bmif}{\bifd}
\safemath{\bmig}{\bigd}
\safemath{\bmih}{\bihd}
\safemath{\bmii}{\biid}
\safemath{\bmij}{\bijd}
\safemath{\bmik}{\bikd}
\safemath{\bmil}{\bild}
\safemath{\bmim}{\bimd}
\safemath{\bmin}{\bind}
\safemath{\bmio}{\biod}
\safemath{\bmip}{\bipd}
\safemath{\bmiq}{\biqd}
\safemath{\bmir}{\bird}
\safemath{\bmis}{\bisd}
\safemath{\bmit}{\bitd}
\safemath{\bmiu}{\biud}
\safemath{\bmiv}{\bivd}
\safemath{\bmiw}{\biwd}
\safemath{\bmix}{\bixd}
\safemath{\bmiy}{\biyd}
\safemath{\bmiz}{\bizd}

\safemath{\bmxi}{\bixid}
\safemath{\bmlambda}{\bilambdad}
\safemath{\bmmu}{\bimud}
\safemath{\bmtheta}{\bithetad}
\safemath{\bmphi}{\biphid}
\safemath{\bmdelta}{\bideltad}

\safemath{\bA}{\mathbf{A}}
\safemath{\bB}{\mathbf{B}}
\safemath{\bC}{\mathbf{C}}
\safemath{\bD}{\mathbf{D}}
\safemath{\bE}{\mathbf{E}}
\safemath{\bF}{\mathbf{F}}
\safemath{\bG}{\mathbf{G}}
\safemath{\bH}{\mathbf{H}}
\safemath{\bI}{\mathbf{I}}
\safemath{\bJ}{\mathbf{J}}
\safemath{\bK}{\mathbf{K}}
\safemath{\bL}{\mathbf{L}}
\safemath{\bM}{\mathbf{M}}
\safemath{\bN}{\mathbf{N}}
\safemath{\bO}{\mathbf{O}}
\safemath{\bP}{\mathbf{P}}
\safemath{\bQ}{\mathbf{Q}}
\safemath{\bR}{\mathbf{R}}
\safemath{\bS}{\mathbf{S}}
\safemath{\bT}{\mathbf{T}}
\safemath{\bU}{\mathbf{U}}
\safemath{\bV}{\mathbf{V}}
\safemath{\bW}{\mathbf{W}}
\safemath{\bX}{\mathbf{X}}
\safemath{\bY}{\mathbf{Y}}
\safemath{\bZ}{\mathbf{Z}}

\safemath{\bZero}{\mathbf{0}}
\safemath{\bOne}{\mathbf{1}}
\safemath{\bDelta}{\mathbf{\Delta}}
\safemath{\bLambda}{\mathbf{\UpLambda}}
\safemath{\bPhi}{\mathbf{\Upphi}}
\safemath{\bSigma}{\mathbf{\Upsigma}}
\safemath{\bOmega}{\mathbf{\Upomega}}
\safemath{\bTheta}{\mathbf{\Uptheta}}

\bmdefine{\biAd}{A}
\bmdefine{\biBd}{B}
\bmdefine{\biCd}{C}
\bmdefine{\biDd}{D}
\bmdefine{\biEd}{E}
\bmdefine{\biFd}{F}
\bmdefine{\biGd}{G}
\bmdefine{\biHd}{H}
\bmdefine{\biId}{I}
\bmdefine{\biJd}{J}
\bmdefine{\biKd}{K}
\bmdefine{\biLd}{L}
\bmdefine{\biMd}{M}
\bmdefine{\biOd}{N}
\bmdefine{\biPd}{O}
\bmdefine{\biQd}{P}
\bmdefine{\biRd}{R}
\bmdefine{\biSd}{S}
\bmdefine{\biTd}{T}
\bmdefine{\biUd}{U}
\bmdefine{\biVd}{V}
\bmdefine{\biWd}{W}
\bmdefine{\biXd}{X}
\bmdefine{\biYd}{Y}
\bmdefine{\biZd}{Z}

\bmdefine{\biDelta}{\Delta}
\bmdefine{\biLambda}{\Lambda}
\bmdefine{\biPhi}{\Phi}
\bmdefine{\biSigma}{\Sigma}
\bmdefine{\biOmega}{\Omega}
\bmdefine{\biTheta}{\Theta}

\safemath{\bimA}{\biAd}
\safemath{\bimB}{\biBd}
\safemath{\bimC}{\biCd}
\safemath{\bimD}{\biDd}
\safemath{\bimE}{\biEd}
\safemath{\bimF}{\biFd}
\safemath{\bimG}{\biGd}
\safemath{\bimH}{\biHd}
\safemath{\bimI}{\biId}
\safemath{\bimJ}{\biJd}
\safemath{\bimK}{\biKd}
\safemath{\bimL}{\biLd}
\safemath{\bimM}{\biMd}
\safemath{\bimN}{\biNd}
\safemath{\bimO}{\biOd}
\safemath{\bimP}{\biPd}
\safemath{\bimQ}{\biQd}
\safemath{\bimR}{\biRd}
\safemath{\bimS}{\biSd}
\safemath{\bimT}{\biTd}
\safemath{\bimU}{\biUd}
\safemath{\bimV}{\biVd}
\safemath{\bimW}{\biWd}
\safemath{\bimX}{\biXd}
\safemath{\bimY}{\biYd}
\safemath{\bimZ}{\biZd}

\safemath{\bimDelta}{\biDelta}
\safemath{\bimLambda}{\biLambda}
\safemath{\bimPhi}{\biPhi}
\safemath{\bimSigma}{\biSigma}
\safemath{\bimOmega}{\biOmega}
\safemath{\bimTheta}{\biTheta}

\safemath{\setA}{\mathcal{A}}
\safemath{\setB}{\mathcal{B}}
\safemath{\setC}{\mathcal{C}}
\safemath{\setD}{\mathcal{D}}
\safemath{\setE}{\mathcal{E}}
\safemath{\setF}{\mathcal{F}}
\safemath{\setG}{\mathcal{G}}
\safemath{\setH}{\mathcal{H}}
\safemath{\setI}{\mathcal{I}}
\safemath{\setJ}{\mathcal{J}}
\safemath{\setK}{\mathcal{K}}
\safemath{\setL}{\mathcal{L}}
\safemath{\setM}{\mathcal{M}}
\safemath{\setN}{\mathcal{N}}
\safemath{\setO}{\mathcal{O}}
\safemath{\setP}{\mathcal{P}}
\safemath{\setQ}{\mathcal{Q}}
\safemath{\setR}{\mathcal{R}}
\safemath{\setS}{\mathcal{S}}
\safemath{\setT}{\mathcal{T}}
\safemath{\setU}{\mathcal{U}}
\safemath{\setV}{\mathcal{V}}
\safemath{\setW}{\mathcal{W}}
\safemath{\setX}{\mathcal{X}}
\safemath{\setY}{\mathcal{Y}}
\safemath{\setZ}{\mathcal{Z}}
\safemath{\emptySet}{\varnothing}

\safemath{\colA}{\mathscr{A}}
\safemath{\colB}{\mathscr{B}}
\safemath{\colC}{\mathscr{C}}
\safemath{\colD}{\mathscr{D}}
\safemath{\colE}{\mathscr{E}}
\safemath{\colF}{\mathscr{F}}
\safemath{\colG}{\mathscr{G}}
\safemath{\colH}{\mathscr{H}}
\safemath{\colI}{\mathscr{I}}
\safemath{\colJ}{\mathscr{J}}
\safemath{\colK}{\mathscr{K}}
\safemath{\colL}{\mathscr{L}}
\safemath{\colM}{\mathscr{M}}
\safemath{\colN}{\mathscr{N}}
\safemath{\colO}{\mathscr{O}}
\safemath{\colP}{\mathscr{P}}
\safemath{\colQ}{\mathscr{Q}}
\safemath{\colR}{\mathscr{R}}
\safemath{\colS}{\mathscr{S}}
\safemath{\colT}{\mathscr{T}}
\safemath{\colU}{\mathscr{U}}
\safemath{\colV}{\mathscr{V}}
\safemath{\colW}{\mathscr{W}}
\safemath{\colX}{\mathscr{X}}
\safemath{\colY}{\mathscr{Y}}
\safemath{\colZ}{\mathscr{Z}}

\safemath{\opA}{\mathbb{A}}
\safemath{\opB}{\mathbb{B}}
\safemath{\opC}{\mathbb{C}}
\safemath{\opD}{\mathbb{D}}
\safemath{\opE}{\mathbb{E}}
\safemath{\opF}{\mathbb{F}}
\safemath{\opG}{\mathbb{G}}
\safemath{\opH}{\mathbb{H}}
\safemath{\opI}{\mathbb{I}}
\safemath{\opJ}{\mathbb{J}}
\safemath{\opK}{\mathbb{K}}
\safemath{\opL}{\mathbb{L}}
\safemath{\opM}{\mathbb{M}}
\safemath{\opN}{\mathbb{N}}
\safemath{\opO}{\mathbb{O}}
\safemath{\opP}{\mathbb{P}}
\safemath{\opQ}{\mathbb{Q}}
\safemath{\opR}{\mathbb{R}}
\safemath{\opS}{\mathbb{S}}
\safemath{\opT}{\mathbb{T}}
\safemath{\opU}{\mathbb{U}}
\safemath{\opV}{\mathbb{V}}
\safemath{\opW}{\mathbb{W}}
\safemath{\opX}{\mathbb{X}}
\safemath{\opY}{\mathbb{Y}}
\safemath{\opZ}{\mathbb{Z}}
\safemath{\opZero}{\mathbb{O}}
\safemath{\identityop}{\opI}

\safemath{\veca}{\bma}
\safemath{\vecb}{\bmb}
\safemath{\vecc}{\bmc}
\safemath{\vecd}{\bmd}
\safemath{\vece}{\bme}
\safemath{\vecf}{\bmf}
\safemath{\vecg}{\bmg}
\safemath{\vech}{\bmh}
\safemath{\veci}{\bmi}
\safemath{\vecj}{\bmj}
\safemath{\veck}{\bmk}
\safemath{\vecl}{\bml}
\safemath{\vecm}{\bmm}
\safemath{\vecn}{\bmn}
\safemath{\veco}{\bmo}
\safemath{\vecp}{\bmp}
\safemath{\vecq}{\bmq}
\safemath{\vecr}{\bmr}
\safemath{\vecs}{\bms}
\safemath{\vect}{\bmt}
\safemath{\vecu}{\bmu}
\safemath{\vecv}{\bmv}
\safemath{\vecw}{\bmw}
\safemath{\vecx}{\bmx}
\safemath{\vecy}{\bmy}
\safemath{\vecz}{\bmz}

\safemath{\veczero}{\bmzero}
\safemath{\vecone}{\bmone}
\safemath{\vecxi}{\bmxi}
\safemath{\veclambda}{\bmlambda}
\safemath{\vecmu}{\bmmu}
\safemath{\vectheta}{\bmtheta}
\safemath{\vecphi}{\bmphi}
\safemath{\vecdelta}{\bmdelta}

\safemath{\matA}{\bA}
\safemath{\matB}{\bB}
\safemath{\matC}{\bC}
\safemath{\matD}{\bD}
\safemath{\matE}{\bE}
\safemath{\matF}{\bF}
\safemath{\matG}{\bG}
\safemath{\matH}{\bH}
\safemath{\matI}{\bI}
\safemath{\matJ}{\bJ}
\safemath{\matK}{\bK}
\safemath{\matL}{\bL}
\safemath{\matM}{\bM}
\safemath{\matN}{\bN}
\safemath{\matO}{\bO}
\safemath{\matP}{\bP}
\safemath{\matQ}{\bQ}
\safemath{\matR}{\bR}
\safemath{\matS}{\bS}
\safemath{\matT}{\bT}
\safemath{\matU}{\bU}
\safemath{\matV}{\bV}
\safemath{\matW}{\bW}
\safemath{\matX}{\bX}
\safemath{\matY}{\bY}
\safemath{\matZ}{\bZ}
\safemath{\matzero}{\bmzero}

\safemath{\matDelta}{\bDelta}
\safemath{\matLambda}{\bLambda}
\safemath{\matPhi}{\bPhi}
\safemath{\matSigma}{\bSigma}
\safemath{\matOmega}{\bOmega}
\safemath{\matTheta}{\bTheta}

\safemath{\matidentity}{\matI}
\safemath{\matone}{\matO}

\safemath{\rnda}{A}
\safemath{\rndb}{B}
\safemath{\rndc}{C}
\safemath{\rndd}{D}
\safemath{\rnde}{E}
\safemath{\rndf}{F}
\safemath{\rndg}{G}
\safemath{\rndh}{H}
\safemath{\rndi}{I}
\safemath{\rndj}{J}
\safemath{\rndk}{K}
\safemath{\rndl}{L}
\safemath{\rndm}{M}
\safemath{\rndn}{N}
\safemath{\rndo}{O}
\safemath{\rndp}{P}
\safemath{\rndq}{Q}
\safemath{\rndr}{R}
\safemath{\rnds}{S}
\safemath{\rndt}{T}
\safemath{\rndu}{U}
\safemath{\rndv}{V}
\safemath{\rndw}{W}
\safemath{\rndx}{X}
\safemath{\rndy}{Y}
\safemath{\rndz}{Z}

\safemath{\rveca}{\bimA}
\safemath{\rvecb}{\bimB}
\safemath{\rvecc}{\bimC}
\safemath{\rvecd}{\bimD}
\safemath{\rvece}{\bimE}
\safemath{\rvecf}{\bimF}
\safemath{\rvecg}{\bimG}
\safemath{\rvech}{\bimH}
\safemath{\rveci}{\bimI}
\safemath{\rvecj}{\bimJ}
\safemath{\rveck}{\bimK}
\safemath{\rvecl}{\bimL}
\safemath{\rvecm}{\bimM}
\safemath{\rvecn}{\bimN}
\safemath{\rveco}{\bomO}
\safemath{\rvecp}{\bimP}
\safemath{\rvecq}{\bimQ}
\safemath{\rvecr}{\bimR}
\safemath{\rvecs}{\bimS}
\safemath{\rvect}{\bimT}
\safemath{\rvecu}{\bimU}
\safemath{\rvecv}{\bimV}
\safemath{\rvecw}{\bimW}
\safemath{\rvecx}{\bimX}
\safemath{\rvecy}{\bimY}
\safemath{\rvecz}{\bimZ}

\safemath{\rvecxi}{\bmxi}
\safemath{\rveclambda}{\bmlambda}
\safemath{\rvecmu}{\bmmu}
\safemath{\rvectheta}{\bmtheta}
\safemath{\rvecphi}{\bmphi}

\safemath{\rmatA}{\bimA}
\safemath{\rmatB}{\bimB}
\safemath{\rmatC}{\bimC}
\safemath{\rmatD}{\bimD}
\safemath{\rmatE}{\bimE}
\safemath{\rmatF}{\bimF}
\safemath{\rmatG}{\bimG}
\safemath{\rmatH}{\bimH}
\safemath{\rmatI}{\bimI}
\safemath{\rmatJ}{\bimJ}
\safemath{\rmatK}{\bimK}
\safemath{\rmatL}{\bimL}
\safemath{\rmatM}{\bimM}
\safemath{\rmatN}{\bimN}
\safemath{\rmatO}{\bimO}
\safemath{\rmatP}{\bimP}
\safemath{\rmatQ}{\bimQ}
\safemath{\rmatR}{\bimR}
\safemath{\rmatS}{\bimS}
\safemath{\rmatT}{\bimT}
\safemath{\rmatU}{\bimU}
\safemath{\rmatV}{\bimV}
\safemath{\rmatW}{\bimW}
\safemath{\rmatX}{\bimX}
\safemath{\rmatY}{\bimY}
\safemath{\rmatZ}{\bimZ}

\safemath{\rmatDelta}{\bimDelta}
\safemath{\rmatLambda}{\bimLambda}
\safemath{\rmatPhi}{\bimPhi}
\safemath{\rmatSigma}{\bimSigma}
\safemath{\rmatOmega}{\bimOmega}
\safemath{\rmatTheta}{\bimTheta}

\usepackage{amssymb}
\usepackage{amsfonts}
\usepackage{mathrsfs}
\usepackage{xspace}
\usepackage{bm}
\usepackage{fancyref}
\usepackage{textcomp}

\usepackage{multirow}
\usepackage{stmaryrd}

\newenvironment{textbmatrix}{	\setlength{\arraycolsep}{2.5pt}%
								\big[\begin{matrix}}{\end{matrix}\big]%
								\raisebox{0.08ex}{\vphantom{M}}}

\def\be{\begin{equation}}
\def\ee{\end{equation}}
\def\een{\nonumber \end{equation}}
\def\mat{\begin{bmatrix}}
\def\emat{\end{bmatrix}}
\def\btm{\begin{textbmatrix}}
\def\etm{\end{textbmatrix}}

\def\ba#1\ea{\begin{align}#1\end{align}}
\def\bas#1\eas{\begin{align*}#1\end{align*}}
\def\bs#1\es{\begin{split}#1\end{split}}
\def\bg#1\eg{\begin{gather}#1\end{gather}}
\def\bml#1\eml{\begin{multline}#1\end{multline}}
\def\bi#1\ei{\begin{itemize}#1\end{itemize}}

\newcommand{\lefto}{\mathopen{}\left}

\DeclareMathOperator*{\argmin}{arg\;min}		
\DeclareMathOperator*{\argmax}{arg\;max}		

\newcommand{\Ex}[2]{\ensuremath{\Exop_{#1}\lefto[#2\right]}} 	

\newcommand{\vecnorm}[1]{\lefto\lVert#1\right\rVert}		

\safemath{\dirac}{\delta}					
\safemath{\krond}{\dirac}					

\safemath{\upto}{\uparrow}
\safemath{\downto}{\downarrow}
\safemath{\iu}{j}							
\safemath{\ev}{\lambda}						
\safemath{\hilseqspace}{l^{2}}				
\newcommand{\banachfunspace}[1]{\setL^{#1}}	
\safemath{\hilfunspace}{\banachfunspace{2}}	

\safemath{\SNR}{\textit{SNR}} 				
\safemath{\PAR}{\textit{PAR}} 				
\safemath{\No}{N_0}							
\safemath{\Es}{E_s}							
\safemath{\Eb}{E_b}							
\safemath{\EbNo}{\frac{\Eb}{\No}}
\safemath{\EsNo}{\frac{\Es}{\No}}

\DeclareMathOperator{\CHop}{\ensuremath{\opH}} 
\safemath{\tvir}{\rndh_{\CHop}}				
\safemath{\tvtf}{\rndl_{\CHop}}				
\safemath{\spf}{\rnds_{\CHop}}				
\safemath{\bff}{H_{\CHop}}					

\safemath{\ircf}{r_{h}}						
\safemath{\tftvcf}{r_{s}}					
\safemath{\tfcf}{r_{l}}						
\safemath{\bfcf}{r_{H}}						

\safemath{\tcorr}{c_h}						
\safemath{\scf}{c_{s}}						
\safemath{\tfcorr}{c_{l}}					
\safemath{\fcorr}{c_{H}}						

\safemath{\mi}{I}							
\safemath{\capacity}{C}						

\safemath{\normal}{\mathcal{N}}			
\safemath{\jpg}{\mathcal{CN}}			
\safemath{\mchain}{\leftrightarrow}		

\safemath{\dB}{\,\mathrm{dB}}
\safemath{\dBm}{\,\mathrm{dBm}}
\safemath{\Hz}{\,\mathrm{Hz}}
\safemath{\kHz}{\,\mathrm{kHz}}
\safemath{\MHz}{\,\mathrm{MHz}}
\safemath{\GHz}{\,\mathrm{GHz}}
\safemath{\s}{\,\mathrm{s}}
\safemath{\ms}{\,\mathrm{ms}}
\safemath{\mus}{\,\mathrm{\text{\textmu}s}}
\safemath{\ns}{\,\mathrm{ns}}
\safemath{\ps}{\,\mathrm{ps}}
\safemath{\meter}{\,\mathrm{m}}
\safemath{\mm}{\,\mathrm{mm}}
\safemath{\cm}{\,\mathrm{cm}}
\safemath{\m}{\,\mathrm{m}}
\safemath{\W}{\,\mathrm{W}}
\safemath{\mW}{\, \mathrm{mW}}
\safemath{\J}{\,\mathrm{J}}
\safemath{\K}{\,\mathrm{K}}
\safemath{\bit}{\,\mathrm{bit}}
\safemath{\nat}{\,\mathrm{nat}}

\safemath{\define}{\triangleq}			

\safemath{\equivalent}{\sim}
\safemath{\distas}{\sim}					
\safemath{\sdiff}{\Delta}				

\safemath{\reals}{\mathbb{R}}
\safemath{\positivereals}{\reals_{+}}
\safemath{\integers}{\mathbb{Z}}
\safemath{\posint}{\integers_{+}}
\safemath{\naturals}{\mathbb{N}}
\safemath{\posnaturals}{\naturals_{+}}
\safemath{\complexset}{\mathbb{C}}
\safemath{\rationals}{\mathbb{Q}}

\newcommand*{\fancyrefapplabelprefix}{app}		
\newcommand*{\fancyrefthmlabelprefix}{thm}		
\newcommand*{\fancyreflemlabelprefix}{lem}		
\newcommand*{\fancyrefcorlabelprefix}{cor}		
\newcommand*{\fancyrefdeflabelprefix}{def}		
\newcommand*{\fancyrefproplabelprefix}{prop}		
\newcommand*{\fancyrefexmpllabelprefix}{exmpl}
\newcommand*{\fancyrefalglabelprefix}{alg}		
\newcommand*{\fancyreftbllabelprefix}{tbl}		

\frefformat{vario}{\fancyrefseclabelprefix}{Section~#1}
\frefformat{vario}{\fancyrefthmlabelprefix}{Theorem~#1}
\frefformat{vario}{\fancyreftbllabelprefix}{Table~#1}
\frefformat{vario}{\fancyreflemlabelprefix}{Lemma~#1}
\frefformat{vario}{\fancyrefcorlabelprefix}{Corollary~#1}
\frefformat{vario}{\fancyrefdeflabelprefix}{Definition~#1}
\frefformat{vario}{\fancyreffiglabelprefix}{Figure~#1}
\frefformat{vario}{\fancyrefapplabelprefix}{Appendix~#1}
\frefformat{vario}{\fancyrefeqlabelprefix}{(#1)}
\frefformat{vario}{\fancyrefproplabelprefix}{Proposition~#1}
\frefformat{vario}{\fancyrefexmpllabelprefix}{Example~#1}
\frefformat{vario}{\fancyrefalglabelprefix}{Algorithm~#1}

 \newtheorem{thm}{Theorem}

 \newtheorem{lem}[thm]{Lemma}
 
\safemath{\dictab}{[\,\dicta\,\,\dictb\,]}

\safemath{\ysig}{\bmy}
\safemath{\ysighat}{\hat{\ysig}}
\safemath{\ysigdim}{M}
\safemath{\xsig}{\bmx}
\safemath{\xsigdim}{N}
\safemath{\nx}{n_x}
\safemath{\zsig}{\bmz}
\safemath{\zsigdim}{\ysigdim}
\safemath{\rsig}{\bmr}
\safemath{\Adict}{\bA}
\safemath{\Adicttilde}{\widetilde{\Adict}}
\safemath{\Adictdim}{\outputdim\times\xsigdim}
\safemath{\avec}{\bma}
\safemath{\avectilde}{\tilde{\avec}}
\safemath{\Bdict}{\bB}
\safemath{\Bdicttilde}{\widetilde{\Bdict}}
\safemath{\Cdict}{\bC}
\safemath{\cvec}{\bmc}
\safemath{\Ddict}{\bD}
\safemath{\Ddictdim}{\ysigdim\times\xsigdim}
\safemath{\dvec}{\bmd}
\safemath{\Ddicttilde}{\widetilde{\bD}}
\safemath{\Bonb}{\bB}
\safemath{\bvec}{\bmb}
\safemath{\Bonbdim}{\ysigdim\times\ysigdim}
\safemath{\noise}{\bmn}
\safemath{\noisedim}{\ysigim}
\safemath{\err}{\bme}
\safemath{\errdim}{\ysigdim}
\safemath{\errset}{\setE}
\safemath{\nerr}{n_e}
\safemath{\delop}{\bP_\errset}
\safemath{\delopc}{\bP_{{\errset}^c}}

\safemath{\cplxi}{\imath}
\safemath{\cplxj}{\jmath}

\safemath{\dict}{\matD}
\safemath{\inputdim}{N}		
\safemath{\outputdim}{M}		
\safemath{\sparsity}{S}	
\safemath{\inputdimA}{{N_a}}	
\safemath{\inputdimB}{{N_b}}	
\safemath{\elemA}{{n_a}}	
\safemath{\elemB}{{n_b}}	
\safemath{\resA}{\matR_a}	
\safemath{\resB}{\matR_b}	
\safemath{\subD}{\matS} 
\safemath{\subA}{\matS_a} 
\safemath{\subB}{\matS_b} 
\safemath{\dicta}{\matA} 	
\safemath{\dictb}{\matB} 	
\safemath{\hollowS}{H}
\safemath{\hollowA}{H_a}
\safemath{\hollowB}{H_b}
\safemath{\cross}{Z}
\safemath{\coh}{\mu_d}			
\safemath{\coha}{\mu_a}			
\safemath{\cohb}{\mu_b}			
\safemath{\mubs}{\nu}	
\safemath{\cohm}{\mu_m} 
\safemath{\dictset}{\setD}	
\safemath{\dictsetp}{\dictset(\coh,\coha,\cohb)}	
\safemath{\dictsetgen}{\dictset_\text{gen}}
\safemath{\dictsetgenp}{\dictsetgen(\coh)}
\safemath{\dictsetonb}{\dictset_\text{onb}}
\safemath{\dictsetonbp}{\dictsetonb(\coh)}

\safemath{\leftside}{U}
\safemath{\rightsideA}{R_a}
\safemath{\rightsideB}{R_b}

\safemath{\indexS}{\setI_S} 

\safemath{\na}{n_a}			
\safemath{\nb}{n_b}			
\safemath{\coeffa}{p_i}	
\safemath{\coeffb}{q_j}	
\safemath{\seta}{\setP}		
\safemath{\setb}{\setQ}     
\safemath{\setw}{\setW}	
\safemath{\setz}{\setZ}	
\safemath{\cola}{\veca}		
\safemath{\colb}{\vecb}		
\safemath{\cold}{\vecd}		
\safemath{\inputvec}{\vecx} 	
\safemath{\error}{\vece}	
\safemath{\noiseout}{\vecz} 	
\safemath{\inputvecel}{x}
\safemath{\inputveca}{\vecx_a}
\safemath{\inputvecb}{\vecx_b}
\safemath{\outputvec}{\vecy}	
\safemath{\lambdamin}{\lambda_{\mathrm{min}}}

\safemath{\elltwo}{\ell_2}
\safemath{\ellone}{\ell_1}
\safemath{\ellzero}{\ell_0}
\safemath{\ellinf}{\ell_\infty}
\safemath{\ellinftilde}{\ell_{\widetilde\infty}}
\safemath{\licard}{Z(\coh,\coha,\cohb)}
\safemath{\xsol}{\hat{x}}
\safemath{\xbord}{x_b}		
\safemath{\xstat}{x_s}		
\safemath{\xstatLone}{\tilde{x}_s}
\safemath{\order}{\mathcal{O}} 
\safemath{\scales}{\Theta} 
\safemath{\ones}{\mathbf{1}} 
\safemath{\zeroes}{\mathbf{0}} 
\safemath{\thlone}{\kappa(\coh,\cohb)} 
\safemath{\constoneA}{\delta} 
\safemath{\constoneB}{\epsilon} 
\safemath{\nlarge}{L}				   
\safemath{\sumlarge}{S_\nlarge}
\safemath{\maxlarger}{P_\nlarge}	   
\safemath{\Pzero}{\textrm{P0}}	
\safemath{\Pone}{\textrm{P1}}
\safemath{\vecfir}{\vecw}			 
\safemath{\vecsec}{\vecz}
\safemath{\elvecfir}{w}              
\safemath{\elvecsec}{z}				 
\safemath{\nlargefir}{n}
\safemath{\normout}{\gamma}
\safemath{\auxfun}{h}
\safemath{\supp}{\textrm{supp}}

\safemath{\indexa}{\ell}
\safemath{\indexb}{r}
\safemath{\indexc}{i}
\safemath{\indexd}{j}

\safemath{\project}{P}

\newtheorem{alg}{Algorithm}
\newtheorem{theorem}{Theorem}

\IEEEoverridecommandlockouts 
\allowdisplaybreaks 
\usepackage{color}

\newcommand{\ejW}[1]{e^{j\Omega #1}}
\newcommand{\proj}[2]{\mathcal{P}_{#1}\!\left(#2\right)}

\def\td{\mathrm{d}}

\newcommand{\gint}[1]{\hat{g}\!\left(#1\right)}
\newcommand{\SVD}[1]{\textit{SVD}\!\left(#1\right)}
\usepackage{bbm}

\def\bSigma{\mathbf{\Sigma}}
\def\CP{\mathrm{CP}}
\def\Unif{\mathrm{Unif}}

\newcommand{\pder}[2]{\frac{\partial#1}{\partial#2} }
\makeatletter
\newcommand{\leqnomode}{\tagsleft@true}
\newcommand{\reqnomode}{\tagsleft@false}
\makeatother

\title{Optimality of the Discrete Fourier Transform for Beamspace Massive MU-MIMO  Communication}
\author{\IEEEauthorblockN{Sueda Taner$^\text{1}$ and Christoph Studer$^\text{2}$}\\[0.02cm]
\IEEEauthorblockA{\small $^\text{1}$School of Electrical and Computer Engineering, Cornell University, Ithaca, NY;  e-mail: st939@cornell.edu\\[0.02cm]
$^\text{2}$Department of Information Technology and Electrical Engineering, ETH Z\"urich, Switzerland; e-mail: studer@ethz.ch}\\[-0.8cm]
\thanks{The work of ST and CS was supported by ComSenTer, one of six centers in JUMP, an SRC program sponsored by DARPA. The work of CS was also supported by an ETH Research Grant  and by the US National Science Foundation (NSF) under grants CNS-1717559 and ECCS-1824379.}
\thanks{The authors would like to thank M.~Gauger, M.~Arnold, and S.~ten~Brink for sharing the real-world channel measurements from~\cite{gauger2020massive}.}
}

\begin{document}

\maketitle


\begin{abstract}
Beamspace processing is an emerging technique to reduce baseband complexity in massive multiuser (MU) multiple-input multiple-output (MIMO) communication systems operating at millimeter-wave (mmWave) and terahertz frequencies.  
The high directionality of wave propagation at such high frequencies ensures that only a small number of transmission paths exist between user equipments and basestation (BS). 
In order to resolve the sparse nature of wave propagation, beamspace processing traditionally computes a spatial discrete Fourier transform (DFT) across a uniform linear antenna array at the BS where each DFT output is associated with a specific beam. 
In this paper, we study optimality conditions of the DFT for sparsity-based beamspace processing with idealistic mmWave channel models and realistic channels. To this end, we propose two algorithms that learn unitary beamspace transforms using an $\boldsymbol\ell^{\bf 4}$-norm-based sparsity measure, and we investigate their optimality theoretically and via simulations.
\end{abstract}



\section{Introduction} 
\label{sec:intro}

Massive multi-user (MU) multiple-input multiple-output (MIMO) and millimeter-Wave (mmWave) as well as terahertz (THz) communication are key technologies for 5G and future wireless systems~\cite{larsson14a,rappaport13a}. 
Since wave propagation at mmWave and THz carrier frequencies is highly directional and experiences a strong path loss, only a small number of dominant transmission paths between each user equipment (UE) and the basestation (BS) antenna array is typically present~\cite{rappaport13a,rappaport_book}.
Hence, by taking a spatial discrete Fourier transform (DFT) across the antenna array (e.g., a uniform linear array) one can convert the antenna-domain system into the so-called \emph{beamspace} in which mmWave and THz channel vectors are sparsified~\cite{mirfarshbafan19a,alkhateeb14,schniter14,deng18,seyedicassp20}.
As demonstrated in~\cite{gao16bs,chen17bscs,SayeedGLOBECOM,mirfarshbafan19a,abdelghany2019precoding,seyedicassp20,mahdavi20beamspace,ma18}, low-complexity baseband algorithms and hardware architectures for channel estimation, detection, and precoding can be designed by exploiting the sparse nature of channel vectors in the beamspace domain. 
It is, however, an open question whether the DFT is indeed the optimal sparsifying transform for modeled as well as real-world mmWave and THz channels. 

\subsection{Contributions and Prior Art}
We formulate an $\ell^4$-norm-based optimization problem for  a complex-valued stochastic data model that enables us to learn beamspace transforms for  mmWave and THz systems. 
In order to solve this optimization problem, we adapt the real-valued matching, stretching, and projection (MSP) algorithm from~\cite{zhai20complete} to our complex-valued and stochastic data model, and we propose an alternative algorithm based on coordinate ascent (CA) via a sequence of Givens rotations. 
We then prove that the DFT is a stationary point of the MSP algorithm  and locally optimal (in the $\ell^4$-norm sense) for the CA algorithm for free-space line-of-sight (LoS) mmWave channels. 
We then show numerical experiments with synthetic and real-world channel vectors in order to demonstrate that learned transforms are able to  improve the performance of sparsity-exploiting baseband algorithms in real-world massive MU-MIMO systems.

Dictionary learning algorithms have been proposed in~\cite{ding3,xie19,huang19,fernandez18,srivastava_SBL,srivastava_SBL-KF} for sparse channel estimation. 
In contrast to these results, we analyze optimality properties of learned dictionaries for mmWave and THz channels and build upon on a recent problem setup that uses a smooth, $\ell^4$-norm-based sparsity measure with orthogonal dictionaries, as proposed recently in~\cite{zhai20complete,shen20a}.
We note that unitary dictionaries are advantageous as they do not alter the noise statistics. 
We furthermore extend the real-valued framework and analysis in~\cite{zhai20complete,shen20a} to the complex case and study optimality conditions of the DFT for beamspace processing in mmWave and THz systems. 
A related approach to our complex-valued unitary dictionary learning setting has been employed recently in~\cite{xue20blind}, which proposes an $\ell^3$-norm-based blind data detection algorithm. In contrast to this result, we study the more general problem of learning a fixed beamspace transform via $\ell^4$-norm maximization that is applicable to any channel realization from the given distribution. 
%
 
\subsection{Notation}
Bold lowercase and uppercase letters represent column vectors and matrices, respectively. %
We use $a_k$ for the $k$th entry of $\veca$, $A_{i,k}$ for the $(i,k)$th entry of $\bA$, and $\veca_k$ for the $k$th column of $\bA$. 
The superscripts $(\cdot)^*,(\cdot)^T$, and $(\cdot)^H$ stand for the matrix conjugate, transpose, and Hermitian, respectively. 
The $N\times M$ all-zeros matrix is $\mathbf{0}_{N\times M}$, the $N\times N$ identity matrix is $\bI_N$, and the $N\times N$ unitary DFT matrix is $\bF_N$.
We denote the set of $N\times N$ orthogonal and unitary matrices with $O(N;\opR)$ and $U(N;\opC)$, respectively.
A complex-valued permutation matrix is defined as a unitary matrix in which every row and column has exactly one non-zero entry; the set of complex-valued permutation matrices of dimension $N\times N$ is $\CP(N)$.
We denote the element-wise multiplication, absolute value, and $r$th power by $\circ$, $|\cdot|$, and $(\cdot)^{\circ r}$, respectively. 
Following~\cite{zhai20complete}, we call $\vecnorm{\bA}_p^p\define\sum_{i,k} |A_{i,k}|^p$ the ``$\ell^p$-norm,'' even though the $p$th root is missing.
The Frobenius norm is $\|\bA\|_F = (\sum_{i,k} |A_{i,k}|^2)^{1/2}$.
Real and imaginary parts are indicated by $\Re\{\cdot\}$ and $\Im\{\cdot\}$. 
The expectation operator is~$\Ex{}{\cdot}$. 
We use $\delta[n]$ to refer to the Kronecker delta function for $n\in\opZ$, such that $\delta[0]=1$ and $\delta[n]=0$ for $n\neq 0$.
We denote the Kronecker product by~$\otimes$.
All complex-valued gradients follow the definitions of~\cite{kreutzdelgado2009complex}. 

\section{Prerequisites}
\label{sec:prereq}
\subsection{System and Channel Model}
\label{sec:sys_model}
We consider a massive MU-MIMO uplink system in  which~$U$ single-antenna UEs transmit data to a BS equipped with $B$ antennas.
Let $\bH\in\opC^{B\times U}$ and $\vecs\in\mathcal{S}^U$ denote the MIMO channel matrix and the data symbols from constellation $\mathcal{S}$, respectively. For narrowband transmission, we can express the receive vector $\vecr\in\opC^B$ in the \textit{antenna domain} as
\begin{align} \label{eq:y_Hs}
    \vecr = \bH \vecs + \vecn,
\end{align}
where $\vecn\in\opC^B$ models circularly-symmetric Gaussian noise.

We focus on wave propagation at mmWave and THz frequencies \cite{rappaport_book}, and assume a sufficiently large distance between the BS and the UEs (or scatterers). For a BS equipped with a uniform linear array (ULA) and $\lambda/2$ antenna spacing, where $\lambda$ is the wavelength, the columns of the MIMO channel matrix~$\bH$ representing the channel between a specific UE and the BS antenna array can be modelled as follows~\cite{tse}:
\begin{align} \label{eq:h_sum}
\vech = \textstyle \sum_{\ell=0}^{L-1} \alpha_\ell\vecp(\omega_\ell), \quad \vecp(\omega) = \big[ 1,e^{j\omega},\dots,e^{j(B-1)\omega} \big]^T\!\!.
\end{align}
Here, $L$ refers to the total number of paths arriving at the antenna array (including a potential line-of-sight path), $\alpha_\ell\in\opC$ is the complex-valued channel gain associated with the $\ell$th propagation path, and $\omega_\ell$ is the angular frequency usually given by the relation $\omega_\ell=\pi\sin(\phi_\ell)$, where $\phi_l$ is the incidence angle of the $\ell$th path to the antenna array.

We obtain the \emph{beamspace} (or angular domain) representation by applying a DFT to the receive vector in \fref{eq:y_Hs} as in~\cite{seyedicassp20,mirfarshbafan19a}, which yields
$\bar{\vecr} = \bF_N\vecr = \bar\bH\vecs + \bar{\vecn}$.
We note that the beamspace representation transforms the superposition of $L$ complex sinusoids in~\fref{eq:h_sum} into the frequency domain, which results in sparse beamspace channel vectors if $L$ is small. 
However, real-world mmWave or THz channel vectors are only approximated by~\fref{eq:h_sum}, due to scattering, diffraction, and system or hardware impairments. 
This key fact motivates us to examine the DFT's optimality for beamspace transforms and to \textit{learn} alternative unitary transforms that exhibit superior sparsifying properties than the widely-used DFT.

\subsection{Problem Formulation}
\label{sec:problem}

Reference~\cite{zhai20complete} recently developed an $\ell^4$-norm-based dictionary learning framework over the orthogonal group $O(N;\opR)$ for a set of given real-valued vectors. 
The intuition behind this framework is that maximizing the $\ell^4$-norm of a matrix over a hypersphere promotes sparsity, i.e., the sparsest points on an $\ell^2$-norm-hypersphere have the smallest $\ell^1$-norm and  largest $\ell^4$-norm~\cite{zhai20complete}.
In what follows, we build upon this insight and consider the complex-valued case with a stochastic data model in order to use this framework for beamspace processing.

Suppose the data samples $\vecy(\Omega)\in\opC^N$ depend on a random variable $\Omega$ with probability density function (PDF) $f_\Omega(\Omega)$. Our goal is to learn a unitary transform that sparsifies these random samples in expectation.
Specifically, we measure the sparsity of these data samples after a unitary transform $\bA\in U(N;\opC)$ via the $\ell^4$-norm using the following: 
\begin{align} \label{eq:gint}
\gint{\bA,\vecy(\Omega)}&\!\define \Ex{\Omega} {\vecnorm{\bA\vecy(\Omega)}_4^4} 
\!=\!\! \int \vecnorm{\bA\vecy(\omega)}_4^4 f_\Omega(\omega) \td\omega.
\end{align}
In order to learn unitary transforms for the stochastic data model, we propose to solve the following optimization problem: 
\begin{align}
\bA = \underset{\tilde\bA}{\text{argmax}}\,\, \gint{\tilde\bA,\vecy(\Omega)}  \, \text{subject to }  \tilde\bA\in U(N;\opC). \tag{O1}\label{eq:obj}
\end{align}
The fundamental properties of the $\ell^4$-norm over the real-valued orthogonal group $O(N;\opR)$ were established in~\cite[Lemmas 5 and 6]{zhai20complete} and can be generalized to the complex case by replacing signed permutation matrices with complex permutation matrices, and standard canonical vectors with the columns of complex permutation matrices. A detailed analysis of the complex case will be provided in an extended version of this paper~\cite{journaltobe}.
One important property is the invariance of the $\ell^4$-norm with respect to complex permutations, i.e., for any $\bC\in\CP(N)$, we have $\vecnorm{\bC\bA}_4^4 = \vecnorm{\bA}_4^4$, so that the solutions to the problem in~\fref{eq:obj} are unique up to a complex permutation.


\section{Learning a Sparsifying Transform}
\label{sec:learning}
We  now propose two algorithms to solve \fref{eq:obj}.
We keep our explanations general while introducing our algorithms, and discuss concrete beamspace applications in Sections \ref{sec:opt_dft} and \ref{sec:results}.

\subsection{MSP: Matching, Stretching, and Projection}
\label{sec:msp}
The proposed MSP algorithm builds upon~\cite [Algorithm 2] {zhai20complete} for our  complex-valued stochastic model instead of using a finite set of real-valued observation samples. 
This adaptation can be interpreted as having an infinitely large set of data vectors that follow a probabilistic distribution. 
In essence, the MSP algorithm performs a projected gradient ascent in the objective~\fref{eq:obj} with an infinite step size. 
In each iteration $t$, we \textit{match} the estimate~$\bA_t$ to the observation~$\vecy(\omega)$, \textit{stretch} all entries of $\bA_t\vecy(\omega)$ with the cubic function in the gradient $|\bA_t\vecy(\omega)|^{\circ2}\circ(\bA_t\vecy(\omega))$, and \textit{project} it back onto the unitary group~\cite{zhai20complete}.
Given the singular value decomposition (SVD) of $\bA$ as $\SVD{\bA} = \bU\bSigma\bV^H$, projection onto the unitary group is accomplished  by
\begin{align}
    \proj{U(N;\opC)}{\bA} \define \underset{\bM\in U(N;\opC)}{\argmin} \vecnorm{\bM-\bA}_F^2 = \bU\bV^H,
\end{align}
where~\cite[Lemma 9]{zhai20complete} is extended to the complex set. The resulting MSP procedure is summarized as follows: 

\vspace{-0.1cm}
\begin{oframed}
\vspace{-0.35cm}
\begin{alg}[MSP]\label{alg:msp}
Initialize $\bA_0\in U(N;\opC)$. For every iteration $t=0,1,\dots,$ until convergence, compute the gradient of the objective with respect to $\bA_t$ as
\begin{align}
&\nabla_{\bA_t} \gint{\bA_t,\vecy(\Omega)} \nonumber\\
&= \int 2\big(|\bA_t\vecy(\omega)|^{\circ2}\circ(\bA_t\vecy(\omega))\big) \vecy(\omega)^H f_\Omega(\omega) \td\omega \label{eq:msp_int}
\end{align}
and project the gradient onto the unitary group
\begin{equation}
\bA_{t+1} = \proj{U(N;\opC)}{\nabla_{\bA_t} \gint{\bA_t,\vecy(\Omega)}}\!. \label{eq:msp_proj}
\end{equation}
\end{alg}
\vspace{-0.35cm}
\end{oframed}
We note that the analysis of real-valued MSP algorithm in~\cite[Propositions 12-15, Theorem 16]{zhai20complete} can be generalized to the unitary case, which we will provide in our journal version~\cite{journaltobe}. 
In order to identify stationary points of~\fref{alg:msp}, we will use the following result:

\begin{lem} \label{lem:UVH}
Let $\bA\in U(N;\opC)$ be a unitary matrix with $\SVD{\bA}=\bU\bSigma\bV^H$ and $\bD\in\opR^{N\times N}$ be a diagonal matrix.
Suppose we have the matrices $\bA_1 = \bD\bA$ and $\bA_2 = \bA\bD$ with $\SVD{\bA_1}=\bU_1\bSigma_1\bV_1^H$ and $\SVD{\bA_2}=\bU_2\bSigma_2\bV_2^H$. Then, 
there exist SVDs for $\bA_1$ and $\bA_2$ such that $\bSigma_1 = \bSigma_2 = \bD$ and $\bA = \bU\bV^H = \bU_1\bV_1^H = \bU_2\bV_2^H$.
\end{lem}
\begin{proof}
Since the left and right singular vectors of full rank matrices are unique up to complex rotations, $\bU\bV^H$ is unique. Since $\bA$ is unitary, one SVD of $\bA$ is $\bU=\bI_N$, $\bSigma = \bI_N$, $\bV=\bA$.
Then $\bA_1=\bD\bA=\bD\bV=\bU_1\bSigma_1\bV_1$, so there exists an SVD of $\bA_1$ such that $\bU_1=\bI_N, \bSigma_1=\bD$, and $\bV_1 = \bA$. Another SVD of $\bA$ is $\bU=\bA$, $\bSigma = \bI_N$, and $\bV=\bI_N$.
Then, $\bA_2=\bA\bD=\bU\bD=\bU_2\bSigma_2\bV_2$, so there exists an SVD of $\bA_2$ such that $\bU_2=\bA, \bSigma_2=\bD$, and $\bV_2 = \bI_N$. Consequently, we have $\bU_1\bV_1^H=\bU_2\bV_2^H = \bA$.
\end{proof}
We now use~\fref{lem:UVH} to determine the following condition when a unitary matrix  is a stationary point of~\fref{alg:msp}:
\begin{lem} \label{lem:msp_st}
Let $\vecy(\Omega)\in\opC^{N}$ represent the stochastic data to be sparsified and $\bA\in U(N;\opC)$ be a dictionary for $\vecy(\Omega)$.
Then, the matrix $\bA$ is a stationary point of \fref{alg:msp} if there exists a diagonal matrix $\bD\in\opR^{N\times N}$ such that $\nabla_\bA \gint{\bA,\vecy(\Omega)} = \bD\bA$ or $\nabla_\bA \gint{\bA,\vecy(\Omega)} = \bA\bD$.
\end{lem}
\begin{proof}
Suppose $\bA_t = \bA$ such that $\nabla_{\bA_t} \gint{\bA_t,\vecy(\Omega)} = \bD\bA_t$ or $\nabla_{\bA_t} \gint{\bA_t,\vecy(\Omega)} = \bA_t\bD$. Then, we have 
\begin{align}
 \bA_{t+1} = \proj{U(N;\opC)}{\nabla_{\bA_t} \gint{\bA_t,\vecy(\Omega)}} = \bA_t,
\end{align} 
where the last equality follows by~\fref{lem:UVH}. 
Since $\bA = \bA_t = \bA_{t+1}$, the matrix $\bA$ is a stationary point of~\fref{alg:msp}.
\end{proof}
In \fref{sec:dft_msp}, we will use~\fref{lem:msp_st} to examine the DFT's optimality for beamspace transforms. 
Proving \emph{local} optimality of a stationary point of~\fref{alg:msp} 
is difficult as it requires an analysis of the Hessian---such results are also missing for the real-valued MSP algorithm in~\cite{zhai20complete}. 
Nonetheless, we are able to study local optimality via the procedure detailed next.

\subsection{CA: Coordinate Ascent}
\label{sec:ca}
We now propose a CA algorithm to find a solution to the optimization problem in~\fref{eq:obj} by directly walking on the Stiefel manifold. 
In each iteration, this method preserves unitarity avoiding a projection on the unitary group altogether; this key property enables us to analyze local optimality.
The algorithm bases on the decomposition of unitary matrices into a set of complex-valued phase shifts and real-valued Givens rotations on pairs of rows as done in~\cite{christoph_givens,golub}.

Let $\bG(i,k,\alpha_{i,k})\in\opC^{N\times N}$ for $i>k$ denote the (real-valued) Givens rotation matrix of the form $G_{i,i}=G_{k,k}=\cos(\alpha_{i,k}),\,G_{i,k}= - G_{k,i}=\sin(\alpha_{i,k}),\,G_{\ell,\ell}=1,\ell\neq i,k$, and $G_{\ell,m}=0$ otherwise.
Multiplying a matrix with $\bG(i,k,\alpha_{i,k})$ from the left amounts to a counterclockwise rotation of $\alpha_{i,k}$ radians in the $(i,k)$ coordinate plane.
Now let us define a phase rotation matrix $\bR(k,\beta_k)\in\opC^{N\times N}$, which is a diagonal matrix with $R_{k,k} = e^{j\beta_k}$ and $R_{\ell,\ell} = 1$ if $\ell\neq k$.
Note that the multiplication of 
Givens and phase rotation matrices with a unitary matrix is still unitary. 
Hence, we can maximize the $\ell^4$-norm by iteratively optimizing over the angles $\alpha_{ik},\beta_i$, and $\beta_k$ while preserving unitarity of the transform at every iteration of the CA algorithm.  
Note that since the $\ell^4$-norm is invariant to complex permutations, in each iteration, we first optimize for $\alpha_{i,k}$ and then over $\beta_i$ and $\beta_k$ accordingly. 
The resulting procedure is summarized as follows: 
\vspace{-0.1cm}
\begin{oframed}
\vspace{-0.35cm}
\begin{alg}[CA]\label{alg:ca}
Initialize $\bA_0 \in U(N;\opC)$. For every iteration $t=0,1,\dots,$ until convergence, and for every $(i,k)$ pair such that $i=1,\dots,N-1$ with $i>k$, find 
\begin{align}
\alpha_{i,k} =  \underset{\tilde\alpha \in [0,\pi/2)}{\argmax} \int \vecnorm{\bG(i,k,\tilde\alpha)\bA_t\vecy(\omega)}_4^4 f_\Omega(\omega) \td\omega, \label{eq:ca_alpha} 
\end{align}
\vspace{-0.3cm}
\begin{align}
&\{\beta_i,\beta_k\} =\underset{\tilde\beta_i,\tilde\beta_k \in [0,2\pi)}{\argmax} \notag \\
&\!\!\!\!\!\int \vecnorm{\bG(i,k,\alpha_{i,k})\bR(i,\tilde\beta_i)\bR(k,\tilde\beta_k)\bA_t\vecy(\omega)}_4^4\! f_\Omega(\omega) \td\omega, \label{eq:ca_beta}
\end{align}	
and apply the update
\begin{align}
\bA_{t+1} &=  \bG(i,k,\alpha_{i,k})\bR(i,\beta_i)\bR(k,\beta_k)\bA_t.
\end{align}
\end{alg}
\vspace{-0.35cm}
\end{oframed}
\vspace{-0.1cm}
We have observed that the results obtained by CA are indistinguishable from those obtained by MSP, while its complexity is typically higher, as we have to  iterate through all $(i,k)$ index pairs \textit{at least once}.  
Rather than using CA in practice, its main advantage is that we can establish local optimality. Here, the first and second derivatives are with respect to single variables (the Givens rotation angles) only, whereas MSP requires the gradient with respect to an $N\times N$ matrix. 
The optimality criteria of~\fref{alg:ca} are as follows:
\begin{lem} \label{lem:ca_opt}
Let $\vecy(\Omega)\in\opC^{N}$ represent the stochastic data to be sparsified and $\bA\in U(N;\opC)$ be a dictionary for $\vecy(\Omega)$. 
Let the matrix $\bG$ as defined in~\fref{sec:ca} and $\vecx(\omega)=\bA\vecy(\omega)$.  
Then, the matrix~$\bA$ is a local maximum of~\fref{alg:ca} if and only if the two following conditions hold for all $(i,k),\, i>k$: 
\begin{align}   
i) \quad &\pder{\int \vecnorm{\bG(i,k,\alpha)\vecx(\omega)}_4^4 f_\Omega(\omega)\td\omega}{\alpha}\bigg|_{\alpha=0}  \nonumber \\
&=2\int ( x_k(-x_i)(x_k^*)^2 + x_k^2x_k^*(-x_i^*) \nonumber  \\
&\qquad\qquad+ x_ix_k(x_i^*)^2 + x_i^2x_i^*x_k^*  )f_\Omega(\omega)\td{\omega} = 0 \label{eq:first_deriv} \\ 
ii) \quad &
\pder{^2\int \vecnorm{\bG(i,k,\alpha)\vecx(\omega) }_4^4f_\Omega(\omega)\td\omega}{\alpha^2}\bigg|_{\alpha=0} \nonumber\\
&= 4 \int \big( 2\Re\{x_k^2(x_i^*)^2\} + 4|x_i|^2|x_k|^2 \nonumber \\
&\qquad\qquad- |x_k|^4 - |x_i|^4 \big)f_\Omega(\omega) \td\omega \label{eq:second_deriv} < 0.
\end{align}
\end{lem}
The proof of~\fref{lem:ca_opt} immediately follows from first- and second-derivative tests~\cite{strang1991calculus}.
We will utilize this lemma when examining DFT's optimality for beamspace in~\fref{sec:dft_ca}. 

\section{Optimality of the DFT}
\label{sec:opt_dft}

We now discuss concrete applications of our algorithms in~\fref{sec:learning} to learn beamspace transforms. 
Since our prime goal is to sparsify channel vectors of mmWave and THz systems as explained in~\fref{sec:sys_model}, 
 we adopt the stochastic data model
\begin{equation} \label{eq:ejw}
    \vecy(\Omega)=\ejW{\vecb}, \qquad \Omega\sim \Unif(0,2\pi),
\end{equation}
where $\vecb = [\,0,\,1,\,\dots,\,B-1\,]^T$, thus $\vecy(\Omega)\in\opC^{B}$.
We adopt this single-path model, which corresponds to free-space propagation with a uniform distribution\footnote{Another stochastic data model would be $e^{j\pi\sin{\Phi}\vecb}$,  $\Phi\sim \Unif(0,2\pi)$, which assumes a uniform distribution over the incidence angle~\cite{tse}. However, we use the model in \fref{eq:ejw} to facilitate our theoretical analysis.} over the angular frequency in~\fref{eq:h_sum}, for simplicity in our derivations. We emphasize that the results in this section also hold for an independent, identically-distributed multipath model as it will be shown in~\cite{journaltobe}.
Based on existing results on beamspace transforms~\cite{mirfarshbafan19a,alkhateeb14,schniter14,deng18,seyedicassp20}, it is known that the DFT is a good candidate to sparsify channel vectors. Thus, we use the DFT to initialize both our MSP and CA algorithms in order to (i) examine their optimality and (ii) find potentially better beamspace transforms. 
The optimality analysis of DFT for both algorithms is detailed next.

\subsection{Optimality Analysis with MSP}
\label{sec:dft_msp}
We prove that the DFT is a stationary point of MSP algorithm for the stochastic data model in~\fref{eq:ejw} with the following result: 
\begin{theorem} 
\label{thm:dft_st}
For the stochastic data model $\vecy(\Omega)$ as given in~\fref{eq:ejw}, 
the DFT matrix $\bF_B$ is a stationary point of~\fref{alg:msp}.
\end{theorem}
\begin{proof}
Let $\bF=\bF_B$ for simplicity. 
Inserting $\bF$ for $\bA_0$, $f_\Omega(\omega)=\frac{1}{2\pi} \mathbbm{1}\{\omega\in (0,2\pi)\}$ in \fref{eq:gint}, we obtain
\begin{equation}
    \gint{\bF,\vecy(\Omega)} = \frac{1}{2\pi} \int_0^{2\pi} \vecnorm{\bF\vecy(\omega)}_4^4 \td\omega.
\end{equation}
To establish that $\bF$ is a stationary point, we utilize~\fref{lem:msp_st} to show that the gradient is equal to a column-scaled DFT matrix, i.e., $\nabla_\bF \gint{\bF,\vecy(\Omega)}= \bF\bD$ for some full-rank diagonal matrix $\bD\in\opC^{N\times N}$ satisfying
\begin{equation}\label{eq:aim_dftopt}
\pder{\gint{\bF,\vecy(\Omega)}}{\vecf_k^*} = D_{k,k}\vecf_k,\, \forall k\in \{0,1,\dots,B-1\}.
\end{equation}
Equivalently, 
our aim is to show that $\pder{\gint{\bF,\vecy(\Omega)}}{F_{i,k}^*} F_{i,k}^*$ is independent of $i$ and only depends on $k$ with
\begin{equation} \label{eq:aim_equiv}
 \pder{\gint{\bF,\vecy(\Omega)}}{F_{i,k}^*} F_{i,k}^*= D_{k,k},\, \forall i,k\in \{0,1,\dots,B-1\}.
\end{equation}
By expanding the gradient, we obtain
\begin{equation} \label{eq:gradF}
\pder{\gint{\bF,\vecy(\Omega)}}{F_{i,k}^*} = \frac{1}{2\pi}\!\sum_{\ell=0}^{B-1}\sum_{m=0}^{B-1}\sum_{n=0}^{B-1} F_{i\ell} F_{im}^* F_{ir} C(\ell,m,n,k),
\end{equation}
where \\[-0.9cm]
\begin{align}
 C(\ell,m,n,k) &= \int_{0}^{2\pi}  e^{j\omega (\ell-m+n-k)} \td\omega \\ 
&= 2\pi \delta[\ell-m+n-k].
\end{align} 
Inserting $F_{i,k} =  e^{\frac{-j 2\pi i k}{B}}$ into~\fref{eq:gradF} and simplifying, we obtain
\begin{align}
&\pder{\gint{\bF,\vecy(\Omega)}}{F_{i,k}^*} \notag\\
&= \sum_{\ell=0}^{B-1}\sum_{m=0}^{B-1}\sum_{n=0}^{B-1} e^{\frac{-j 2\pi i(\ell-m+n)}{B}}  \delta[\ell-m+n-k]\\
&= e^{\frac{-j 2\pi ik}{B}} \sum_{\ell=0}^{B-1}\sum_{m=0}^{B-1}\sum_{n=0}^{B-1} \delta[\ell-m+n-k]. \label{eq:gradient_inFik}
\end{align}
Finally, inserting~\fref{eq:gradient_inFik} into \fref{eq:aim_equiv} 
gives that $\pder{\gint{\bF,\vecy(\Omega)}}{F_{i,k}^*} F_{i,k}^*$ indeed only depends on $k$  as 
\begin{align}
\pder{\gint{\bF,\vecy(\Omega)}}{F_{i,k}^*} F_{i,k}^* &=  \sum_{\ell=0}^{B-1}\sum_{m=0}^{B-1}\sum_{n=0}^{B-1} \delta[\ell-m+n-k]\\
& = D_{k,k},\forall i,k\in \{0,1,\dots,B-1\}.
\end{align}
Hence, we have that $\nabla_{\bF_B} \gint{\bF_B,\bY}= \bF\bD$ and by~\fref{lem:msp_st}, the DFT matrix $\bF_B$ is a stationary point of MSP algorithm.
\end{proof}
Note that this analysis does not establish whether the DFT is a saddle point or a local maximum of~\fref{eq:obj}. Fortunately, we can use the proposed  CA algorithm to reach this conclusion.

\subsection{Optimality Analysis with CA}
\label{sec:dft_ca}

We now use the following result to establish that the DFT is a local maximum of~\fref{alg:ca} for the model $\vecy(\Omega)$ in~\fref{eq:ejw}.

\begin{theorem} 
\label{thm:dft_opt}
For the stochastic data model $\vecy(\Omega)$ as given in~\fref{eq:ejw}, the DFT matrix $\bF_B$ is a local maximum of~\fref{alg:ca}.
\end{theorem}
\begin{proof}

Let $\vecx(\omega)= \bF_B \vecy(\omega) =\bF_B  e^{j\omega\vecb}$. Then 
\begin{align}
x_k &= \textstyle \sum_{n=0}^{B-1} e^{j\omega n} e^{-j\frac{2\pi n}{B}k} = \sum_{n=0}^{B-1} e^{j(\omega -\frac{2\pi }{B}k)n} \label{eq:xk}
\end{align}
Inserting $x_k$ from~\fref{eq:xk} into~\fref{eq:first_deriv} and~\fref{eq:second_deriv} followed by a sequence of tedious algebraic simplifications give (i) $\pder{\int_0^{2\pi}\vecnorm{\bG(i,k,\alpha)\bF_B\vecy(\omega)}_4^4\td\omega}{\alpha}\big|_{\alpha=0} = 0$ and (ii) $\pder{^2\int_0^{2\pi}\vecnorm{\bG(i,k,\alpha)\bF_B\vecy(\omega)}_4^4\td\omega}{\alpha^2}\big|_{\alpha=0} < 0 $.
Therefore, by~\fref{lem:ca_opt}, the DFT matrix $\bF_B$ is a local maximum of~\fref{alg:ca}.
\end{proof}
Although we do not know whether the DFT is a global maximum of the CA algorithm, we have observed no solution, upon perturbations and random initialization, that reaches a higher objective in~\fref{eq:gint} for the stochastic data model $\vecy(\Omega)$ in~\fref{eq:ejw}. 
Establishing global optimality is left for future work.  
We would like to remark that the optimality of the DFT is not obvious, since we will show in~\cite{journaltobe} that any known type of the discrete cosine transform (DCT) is not $\ell^4$-norm optimal.
\graphicspath{ {./figures/} }
\newcommand{\figsize}{0.42}

\section{Numerical Results}
\label{sec:results}

We now utilize~\fref{alg:msp} to learn a beamspace transform from a given,  finite-size data set of observation samples for two experiments: (i) a synthetic system model and (ii) a real-world example. 
We emphasize that in these cases, optimality of the DFT for sparsifying channel vectors is no longer guaranteed. 
We also note that for a finite-size data set, our MSP algorithm is essentially equivalent to $\vecy$ having a uniform probability mass function (PMF) over the given samples, which reduces our method to the complex equivalent of~\cite[Algorithm 2]{zhai20complete}.

In both of our experiments, we split the channel matrices into training and test sets; we use the columns of the training set to learn a beamspace transform, then we measure the $\ell^4$-norm and simulate the uncoded bit error rate (BER) with respect to signal-to-noise ratio (SNR) on the test set.
In our simulations, we use the following sparsity exploiting algorithms: (i) Beamspace channel estimation (BEACHES)  from~\cite{mirfarshbafan19a} and (ii) the beamspace largest-entry (LE) data detector from~\cite{seyedicassp20} with density coefficient $0.125$, which operates on a $(B/8)\times U$-sized channel matrix by picking the largest entries to reduce complexity. 
As a baseline method, we also include an antenna-domain linear least-squares minimum mean-square error (LMMSE) data detector combined with BEACHES.

\begin{figure}[tp]
\centering
\includegraphics[width=.72\columnwidth]{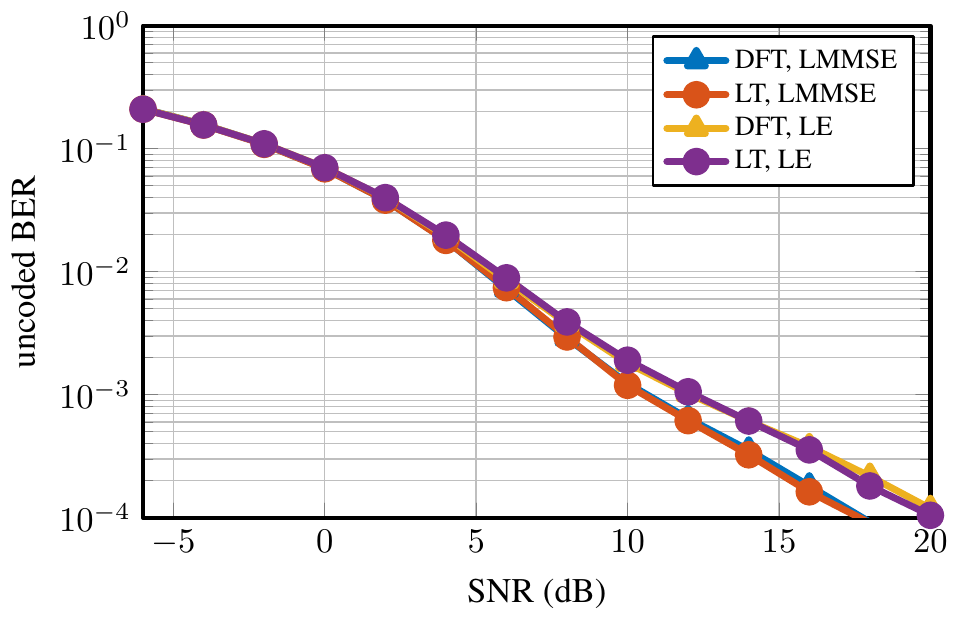}
\vspace{-0.4cm}
\caption{Uncoded BER for beamspace processing with the DFT or learned transform (LT) using synthetic QuaDRiGa channels with LMMSE and  the sparsity-exploiting LE detector, $B=256$ BS antennas, and $U=16$ UEs.}
\label{fig:ber_quad}
\vspace{-0.3cm}
\end{figure}
\begin{figure}
\centering
\includegraphics[width=.72\columnwidth]{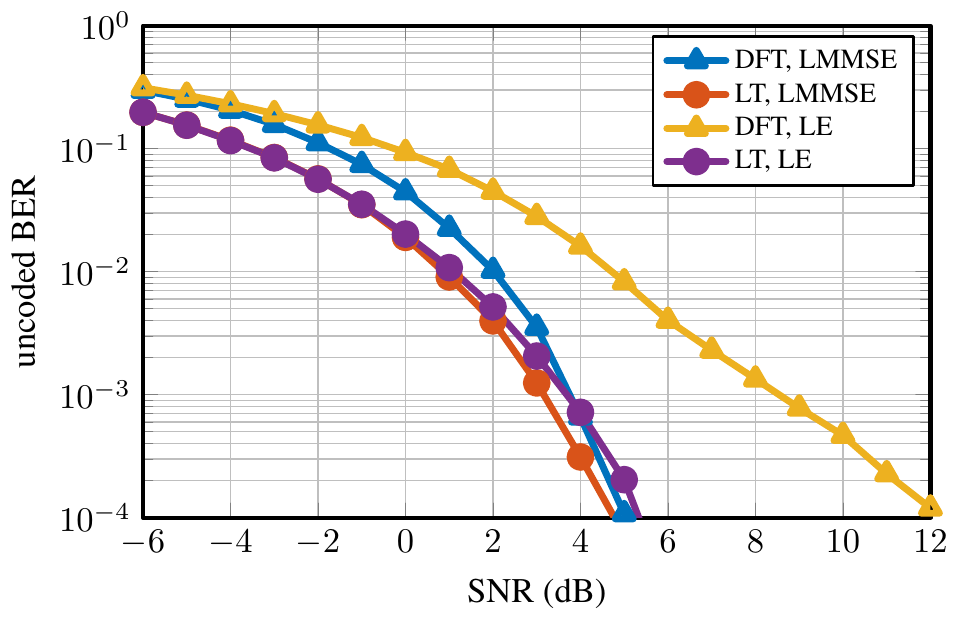}
\vspace{-0.4cm}
\caption{Uncoded BER  for beamspace processing with the DFT or learned transform (LT) using measured channel vectors~\cite{gauger2020massive} with LMMSE and the LE detector, $B=64$ BS antennas (with eight malfunctioning), and $U=1$ UE.}
\label{fig:ber_ctw}
\vspace{-0.1cm}
\end{figure}

\subsection{Synthetic Channel Vectors}

We first simulate line-of-sight (LoS) channel conditions using the QuaDRiGa mmMAGIC UMi model~\cite{QuaDRiGa}, which includes  multipath scattering,  at a carrier frequency of 60\,GHz with a ULA having $\lambda/2$ antenna spacing.
We generate channel matrices for a mmWave massive MIMO system with $B=256$ BS antennas and $U = 16$ single-antenna UEs.
The UEs are placed randomly in a $120^\circ$ circular sector around the BS between a distance of 10\,m and 110\,m, and we assume a minimum UE separation of  $1^\circ$.
We add BS-side power control so that the UE with highest received power has at most $6$\,dB more power than the weakest UE.
We show the BER results for this channel vector set using the DFT and the learned transform (LT) in~\fref{fig:ber_quad}. 
We observe that the LT has only a slight advantage in BER compared to the DFT under the same detector; this ``advantage'' is due to the sparsity-exploiting channel estimation for slightly more sparse channels in beamspace---here, the $\ell^4$-norm of the test set in beamspace domain with LT was only $18\%$ higher than that of the DFT, which can be interpreted as only an approximately $2\%$ higher magnitude in the signal's peaks. 
This result demonstrates that the DFT is (i) no longer optimal but (ii) remains to be an excellent sparsifier for simulated mmWave LoS channels with multipath components. Consequently, it is not worth learning another beamspace transform, which is in congruence with our proof of optimality of DFT for the simple model used in~\fref{eq:ejw}.

\subsection{Real-World Measured Vectors}
We now show results for \emph{measured} channel vectors provided for the IEEE Communications Theory Workshop Localization Competition~\cite{gauger2020massive}.
These channel measurements are based on single-UE transmission to a BS with an $8\times 8$ square antenna array with $\lambda/2$ spacing  at a carrier frequency of 1.27\,GHz. Eight BS antennas were malfunctioning and their output was excluded from the dataset. 
For beamspace processing with rectangular arrays, one would typically apply a two-dimensional DFT on this data as follows: 
Zero-pad for the malfunctioning antennas, vectorize the data to have vectors of size $64$, then multiply with $\bF_8\otimes\bF_8$.
Instead, we use these 64-sized vectors as a training set to learn a beamspace transform.
We show the BER results for this channel vector set using the DFT and the learned transform (LT) in~\fref{fig:ber_ctw}. 
For the LMMSE detector, we observe that the LT can achieve a target BER of $0.1$\% with approximately 1\,dB smaller SNR than the DFT as a result of the sparsity-exploiting channel estimation---here, the $\ell^4$-norm of the test set in the beamspace domain with the LT was up to $4\times$ higher than that of the DFT.
However, for the sparsity-exploiting LE detector, we observe that the LT can achieve the same target BER with approximately 5\,dB smaller SNR than the DFT, allowing the performance of LE to be comparable to antenna-domain LMMSE, as a result of the enhanced sparsity.
This result demonstrates that, although the DFT is well-suited for beamspace processing under idealistic LoS channel conditions, learning new beamspace transforms enables significant improvements for real-world channels and communication systems that suffer from hardware impairments.


\section{Conclusion}
\label{sec:conclusion}

In this paper, we have formulated an optimization problem to learn unitary dictionaries for a complex stochastic model by generalizing the real-valued dictionary learning problem in~\cite{zhai20complete}.
We have proposed two algorithms for this optimization problem: (i) a projected gradient ascent-based algorithm adapted from~\cite[Algorithm 2]{zhai20complete} and a novel coordinate ascent algorithm that avoids projection onto the unitary group. 
We have used the latter algorithm to establish local optimality of the DFT for a free-space mmWave/THz LoS channel model. 
We have used synthetic results to demonstrate that the DFT performs well for idealistic mmWave channel models, but can be improved significantly for real-world measurements with non-ideal hardware using a learned beamspace transform. 

We will show more aspects of our derivations and results in the journal version of this paper~\cite{journaltobe}.
Although our focus was on mmWave and THz communication systems, our algorithms are applicable to more general dictionary learning problems, potentially including other stochastic data models or optimality claims, which leads to many avenues for future work.

\let\oldthebibliography\thebibliography
\let\endoldthebibliography\endthebibliography
\renewenvironment{thebibliography}[1]{
  \begin{oldthebibliography}{#1}
    \setlength{\itemsep}{0.30em}
    \setlength{\parskip}{0em}
}
{
  \end{oldthebibliography}
}
	 
\balance
\bibliographystyle{IEEEtran}

\vspace{2cm}

\linespread{1}

\bibliography{bib/IEEEabrv,bib/confs-jrnls,bib/publishers,bib/REFs,bibfile,vipbib}

\balance
	
\end{document}